\documentclass[runningheads,final]{llncs}

\usepackage[colorlinks,linkcolor={blue},citecolor={blue},urlcolor={red},breaklinks=true,final]{hyperref}
\usepackage[nosumlimits,nointlimits,nonamelimits]{amsmath}

\let\oldmathop\mathop
\renewcommand{\mathop}[1]{\oldmathop{#1}\nolimits}
\makeatletter

\let\c@lemma\c@theorem
\let\c@corollary\c@theorem
\let\c@proposition\c@theorem
\let\c@definition\c@theorem
\let\c@example\c@theorem
\let\c@remark\c@theorem

\@addtoreset{theorem}{section}
\makeatother

\spnewtheorem{defn}[theorem]{Definition}{\bfseries}{\upshape}
\usepackage[T1]{fontenc}
\usepackage{extra-packages}
\usepackage{extra-macros}

\newcommand{\pfin}{\ftP_\omega}
\newcommand{\dist}{\mathcal{D}}
\newcommand{\neighb}{\mathcal{N}}
\newcommand{\mon}{\mathcal{M}}
\newcommand{\filtfun}{\mathcal{F}}
\newcommand{\ultra}{\mathcal{U}}
\newcommand{\nat}{\mathbb{N}}
\newcommand{\integers}{\mathbb{Z}}

\renewcommand{\c}{\colon}
\renewcommand{\emptyset}{\varnothing}

\renewcommand{\ftF}{F}
\renewcommand{\ftS}{S}
\renewcommand{\ftT}{T}

\usepackage{bm}
\renewcommand{\mS}{{\bm{\mathsf{S}}}}
\renewcommand{\mT}{{\bm{\mathsf{T}}}}
\renewcommand{\mP}{{\bm{\mathsf{P}}}}
\renewcommand{\mF}{{\bm{\mathsf{F}}}}
\renewcommand{\mM}{{\bm{\mathsf{M}}}}
\renewcommand{\mZ}{{\bm{\mathsf{1}}}}

\renewcommand{\ftbF}{{\bar{F}}}

\date{}
\begin{document}
\allowdisplaybreaks

\author{
        Sergey Goncharov\inst{1}\orcidID{0000-0001-6924-8766} \and
        Dirk Hofmann\inst{2}\orcidID{0000-0002-1082-6135} \and
        Pedro Nora\inst{3}\orcidID{0000-0001-8581-0675} \and
        Lutz Schr\"oder\inst{4}\orcidID{0000-0002-3146-5906} \and
        Paul~Wild\inst{4}\orcidID{0000-0001-9796-9675}
        }
\authorrunning{S. Goncharov et al.}
\institute{
         University of Birmingham, England \and
         CIDMA, University of Aveiro, Portugal \and
         Universität Duisburg-Essen, Germany \and
         Friedrich-Alexander-Universität Erlangen-Nürnberg, Germany
        }
\title{The Only Distributive Law Over the Powerset Monad Is the One You Know}
\maketitle

\begin{abstract}
Distributive laws of set functors over the powerset monad (also known as Kleisli laws for the powerset monad)  are well-known to be in one-to-one correspondence with extensions of set functors to functors on  the category of sets and relations. We study the question of existence and uniqueness of such  distributive laws.  Our main result entails that an accessible set functor admits a distributive law over the powerset monad if and only if it preserves weak pullbacks, in which case the so-called power law (which induces the Barr extension) is the unique one. Furthermore, we show that the powerset functor admits exactly three distributive laws over the powerset monad, revealing that uniqueness may fail for non-accessible functors.
\end{abstract}

\section{Introduction}

Distributive laws between functors and monads play a central role in
category theory and its applications to semantics, algebra, and
coalgebra. In particular, distributive laws of (endo)functors over the
powerset monad -- also known as \emph{Kleisli laws} for the powerset
monad -- are equivalent to extension of functors from the category of
sets and functions to the Kleisli category of the powerset monad,
equivalently to the category of sets and relations. Such extensions
are fundamental in relational semantics, coalgebraic modal
logic~\cite{KUPKE20115070}, and automata
theory~\cite{HasuoJacobsEtAl07}, where nondeterminism is modelled via
the powerset monad.  Correspondingly, the existence of extensions (or
distributive laws) has been studied
extensively~\cite{Bar70,Jac04,HasuoJacobsEtAl07}; in particular, it is
known that weak-pullback-preserving functors admit a canonical
extension to relations, the so-called Barr extension~\cite{Bar70}. On
the other hand, the question of \emph{uniqueness} of extensions has
received comparatively less attention.

Our main contribution in the present work is a systematic study of the
existence and uniqueness of Kleisli laws for the powerset monad. Our
first main result shows that a large and important class of functors
-- which strictly includes the accessible functors -- admit a
distributive law over the powerset monad only if they preserve weak
pullbacks (as noted above, the reverse implication holds for all
functors). Moreover, in this case the distributive law is unique and
coincides with the well-known \emph{power
  law}~\cite{HasuoJacobsEtAl07}, which induces the Barr extension to
relations. This provides a clean conceptual explanation for the
ubiquity of the Barr extension in coalgebraic semantics.

Our second main result demonstrates that uniqueness may fail outside
the accessible setting. Concretely, we show that the powerset functor
itself admits exactly three distinct distributive laws over the
powerset monad (note that this contrasts sharply with distributive
laws of the powerset \emph{monad} over itself, which do not
exist~\cite{KLIN2018261}). This example shows that preservation of
weak pullbacks, while ensuring existence of a Kleisli distributive law
for the powerset monad, does not by itself suffice to guarantee
uniqueness, and that accessibility plays a crucial role in our
characterization.

The importance of distributive laws for the powerset monad has been highlighted in a number of foundational works, including generic treatments of distributive laws and their applications to semantics. In particular, Jacobs and others~\cite{Jac04,HasuoJacobsEtAl07} have emphasized the role of relational extensions in coalgebraic logic, modal semantics, and notions of behavioural equivalence. Our results refine and extend this line of research by identifying precise conditions under which such liftings are unique.

\paragraph{Related work.} It is well-known that the Barr extension of
a functor is functorial if and only if the functor preserves weak
pullbacks~\cite{Bar70,Trn80}. The `only if' direction of this result
is complemented by our results, which show in particular that it
suffices to assume any extension (rather than specifically the Barr
extension), provided that the functor is accessible (or in fact
satisfies an even weaker condition).

Several recent works establish \emph{no-go theorems} for distributive laws in various settings. Zwart and Marsden~\cite{ZM22} study distributive laws between monads using algebraic presentations and derive general impossibility results. Kupke and Seidel~\cite{KS24} prove no-go theorems for distributive laws from comonads to monads, including the result that there is no distributive law of a polynomial functor over the non-empty powerset monad. These results have been further extended~\cite{PD25}, again using algebraic techniques. In contrast, our approach is not based on algebraic presentations, but rather on categorical properties such as accessibility and weak pullback preservation.

Closer to our setting, Kurz and Velebil~\cite{KV16} study extensions of functors to relations and explicitly leave open the question of uniqueness. Our results provide a negative answer to this question in general, while identifying a broad class of functors for which uniqueness does hold.

B\'ilkov\'a et al~\cite{BKP+11} prove a uniqueness result for
extensions of functors from the category of preorders and monotone
maps to the corresponding category of distributors, under a local
monotonicity assumption. This result has since been generalized to
$\V$-categories~\cite{BKP+13}. It is shown that a functor admits an
extension if and only if it preserves weak pullbacks; this result does
not contradict our results as the overall setting in \emph{op.~cit.}
entails that only locally monotone extensions are considered. 

Schubert and Seal~\cite{SS08,Sea09} use monad morphisms to construct \emph{lax extensions} to the category of sets and relations of set monads that factor through the category of complete lattices and supremum-preserving maps, where lax extensions weaken the functoriality requirement to only one inclusion.

\section{Preliminaries}

We assume basic familiarity with category theory.  Throughout, we work
in the categories $\SET$ of sets and functions and $\REL$ of sets and
relations.  We are interested in endofunctors on $\SET$ and their
extensions to $\REL$; all references to functors and extensions are to
be understood in this sense, unless stated otherwise.

\emph{Accessible functors.}  Given a regular cardinal $\lambda$, a
functor $\ftF$ is \emph{$\lambda$-accessible}~\cite{MP89} if it
preserves $\lambda$-filtered colimits; equivalently, if for every set
$X$,
\begin{equation}
        \label{p:7}
        \ftF X = \bigcup\{\ftF i_A [\ftF A]  \mid |A| < \lambda\},
\end{equation}
where $i_A \colon A \rightarrowtail X$ is the inclusion of $A$ into
$X$.  In particular, $\omega$-accessible functors are called
\emph{finitary}.  A functor is \emph{accessible} if it is
$\lambda$-accessible for some regular cardinal $\lambda$.  Every
functor $\ftF$ has a canonical $\lambda$-accessible subfunctor
$\ftF_\lambda$ whose action on sets is given by the right-hand side of
\eqref{p:7}; we refer to $\ftF_\omega$ as the \emph{finitary part}
of~$\ftF$.

\begin{example}
  \label{p:28}
  The following functors are accessible:
  \begin{enumerate}
  \item Every polynomial functor ($\ftF X=\sum_{i\in I} A_i\times X^{B_i}$
    for sets $I, A_i, B_i$).
  \item The finite powerset functor~$\pfin$, and every subfunctor
    $\ftP_n$ of $\pfin$ given by restricting to cardinality at
    most~$n$ ($\ftP_nX=\{A\in\pfin X\mid |A|\le n\}$).
  \item The discrete distribution functor~$\dist$, given on sets~$X$
    by $\dist X=\{\mu\colon X\to[0,1]\mid \sum_{x\in X}\mu(x)=1\}$.
  \item The monoid-valued functor $M^{(-)}$, given on sets~$X$ by
    \begin{equation*}
      M^{(X)}=\{f\colon X\to M\mid f(x)=0\text{ for all but finitely
        many~$x$}\},
    \end{equation*}
    for a commutative monoid $M=(M,+,0)$. Instances include the
    multiset functor $\nat^{(-)}$ (free commutative monoids) and the
    functor $\integers^{(-)}$ (free abelian groups).
  \end{enumerate}
\end{example}
\begin{example}
  \label{p:55}
  The following functors are \emph{not} accessible:
  \begin{enumerate}
  \item The (full) powerset functor~$\ftP$.
  \item The neighbourhood functor~$\neighb$, given on sets~$X$ by
    $\neighb X=2^{(2^X)}$ and on maps~$f$ by $\neighb f=2^{(2^f)}$.
    That is,~$\neighb$ is double contravariant powerset, and we regard
    elements of $\neighb X$ as systems of subsets of~$X$, termed
    \emph{neighbourhood systems}.
  \item The monotone neighbourhood functor~$\mon$, given as a
    subfunctor of~$\neighb$ by restricting to neighbourhood systems
    that are \emph{monotone}, i.e.~upwards closed under subset
    inclusion.
  \item The filter functor~$\filtfun$, given as a subfunctor of~$\mon$
    by restricting to monotone neighbourhood systems that are \emph{filters},
    that is, closed under finite intersections.
  \item The ultrafilter functor~$\ultra$, given as a subfunctor
    of~$\filtfun$ by restricting to \emph{ultrafilters}, that is, filters
    that are \emph{non-trivial}, i.e.~do not contain~$\emptyset$, and
    maximal.
  \end{enumerate}
\end{example}

\emph{Relations, extensions and Kleisli laws.}  A relation from $X$ to
$Y$, i.e.~a subset $r\subseteq X\times Y$, will be denoted by
$r \colon X \relto Y$ and its converse by $r^\circ \colon Y \relto X$,
and given relations $r \colon X \relto Y$ and $s \colon Y \relto Z$,
we denote their composite by $s \cdot r\colon X\relto Z$ (explicitly,
$s\cdot r=\{(x,z)\mid \exists y.\,(x,y)\in r\land (y,z)\in s\}$).
If $r\colon X \relto Y$ is a relation and $A\subseteq X$, then the
\emph{relational image} of $A$ under $r$ is
$r[A] = \{y\in Y \mid \exists x. (x,y)\in r\}$.
We often write~$\le$ for subset inclusion among relations, and
$x\mathrel{r}y$ for $(x,y)\in r$. A functor
\(\extF \colon \REL \to \REL\) is called an \emph{extension} of a
functor \(\ftF \colon \SET\to \SET\) if it makes the following diagram
commute, where the vertical arrows denote the graph functor:
\begin{center}
 	\begin{tikzcd}
                \REL & \REL \\
 	        \SET & \SET.
 	        \ar[from=1-1,to=1-2,"\extF "]
 	        \ar[from=2-1,to=1-1]
 	        \ar[from=2-1,to=2-2, "\ftF"']
 	        \ar[from=2-2,to=1-2]
 	\end{tikzcd}
\end{center}
In other words, an extension of $\ftF$ consists of maps
\begin{equation*}
  (r\colon X \relto Y)
  \longmapsto 
  (\extF  r \colon \ftF X \relto \ftF Y)
\end{equation*}
such that for all relations $r \colon X \relto Y$,
$s \colon Y \relto Z$, and every function $f \colon X \to Y$,
\begin{itemize}[wide]
        \item
                $\extF  s\cdot\extF  r=\extF  (s\cdot r)$,
        \item
                $\extF  f = \ftF f$,
\end{itemize}
where functions are interpreted as relations via their graph.
An extension $\extF $ is \emph{locally monotone} if $r \leq r'$
entails $\extF r \leq \extF r'$.  It is well-known that extensions of
a functor~$\ftF$ correspond precisely to Kleisli distributive laws
of~$\ftF$ over the powerset monad \cite{Bec69,Str72}.  The powerset
monad $(\ftP,\eta,\mu)$ consists of the powerset functor
$\ftP \colon \SET \to \SET$ and two natural transformations, the unit
$\eta \colon \ftId \to \ftP$ given by $\eta(x)=\{x\}$, and the
multiplication $\mu \colon \ftP\ftP \to \ftP$ which takes unions of
systems of sets.  We recall that the category $\REL$ is isomorphic to
the Kleisli category $\SET_\mP$, which has all sets as objects, and
maps $X\to\ftP Y$ as morphisms $X\to Y$.  We write
\(r^\sharp \colon X \to \ftP Y\) for the function corresponding to the
relation \(r\) under the isomorphism \(\REL \simeq \SET_\mP\). A
\emph{distributive law} of a functor $\ftF \colon \SET \to \SET$ over
the powerset monad is a natural transformation
$\sigma \colon \ftF\ftP \to \ftP\ftF$ that commutes with the unit and
the multiplication:
\begin{equation*}
\begin{array}{c c}
\begin{tikzcd}
        \ftF & \ftF\ftP  \\
          & \ftP F
        \ar[from=1-1, to=1-2, "\ftF\eta"]
        \ar[from=1-1, to=2-2, "\eta \ftF"']
        \ar[from=1-2, to=2-2, "\sigma"]
\end{tikzcd}
&
\hspace{3em}
\begin{tikzcd}
        \ftF\ftP^2   & & \ftF\ftP \\
        \ftP\ftF\ftP & \ftP^2\ftF & \ftP\ftF.
        \ar[from=1-1,to=1-3,"\ftF\mu"]
        \ar[from=1-1,to=2-1,"\sigma\ftP"']
        \ar[from=1-3,to=2-3,"\sigma"]
        \ar[from=2-1,to=2-2,"\ftP\sigma"']
        \ar[from=2-2,to=2-3,"\mu\ftF"']
\end{tikzcd}
\end{array}
\end{equation*}
The extension~$\extF$ of~$\ftF$ induced by~$\sigma$ is best described
in terms of Kleisli morphisms, namely as mapping a Kleisli morphism
$f\colon X\to\ftP Y$ to the Kleisli morphism
$\sigma_Y\cdot\ftF f\colon \ftF X\to \ftP\ftF Y$. Conversely, an
extension~$\extF$ of~$\ftF$ induces a distributive law~$\sigma$ as
follows: The identity map $\ftP X\to\ftP X$ corresponds to a relation
$\owns\colon \ftP X\relto X$; we then obtain a relation
$\extF{\owns}\colon F\ftP X\relto FX$, which corresponds to a Kleisli
morphism $\sigma_X\colon F\ftP X\to \ftP FX$.

Our main goal is to show that an accessible functor admits at most one
extension, which then coincides with the popular Barr
extension~\cite{Bar70}. Generally, we define the \emph{Barr
  extension}~$\ftbF$ of a functor~$\ftF$ as follows. Given a relation
$r \colon X \relto Y$ and a factorization $r=g\cdot f^\circ$ for maps
$X\xleftarrow{f}Z\xrightarrow{g} Y$, put
$\ftbF r = \ftF g \cdot {(\ftF f)}^\circ$.
One can show that this assignment is well-defined, i.e.~independent of
the factorization of $r$, and $r$ admits a \emph{canonical
  factorization} $r=\pi_2 \cdot \pi_1^\circ$ arising from the span
$X \xleftarrow{\pi_1} R \xrightarrow{\pi_2} Y$ where~$R$ is just the
relation~$r$, used as an object of~$\SET$, and~$\pi_1,\pi_2$ are the
projection functions. As indicated previously, the Barr extension may
fail to be an actual extension. Indeed, it is well-known that this
construction yields an extension that is locally monotone iff the
functor preserves weak pullbacks~\cite{Bar70,Trn80}.

\begin{example}
  \label{p:56}
  All functors from \Cref{p:55} admit monad structures.
  We already discussed the powerset monad $\mP$.
  For each of the other functors $F\in\{\ultra,\filtfun,\mon,\neighb\}$, the unit is the map $e^F\colon x \mapsto \{A \subseteq X \mid x \in A\}$ assigning to each element its principal ultrafilter, while the respective multiplications are given by $m^F\colon\Phi \mapsto \{A \subseteq X \mid \{U \in FX \mid A\in U\} \in \Phi\}$ 
  We generally use bold letters to denote these monads; e.g. $\mM = (\mon,e^\mon,m^\mon)$ is the monotone neighbourhood monad, and similarly for the other functors.
\end{example}
  
\paragraph{Weak Pullbacks} A commutative square
\begin{equation*}
\begin{tikzcd}[column sep=normal, row sep=normal]
A\rar["p"]\dar["q"'] & B\dar["f"]\\
C\rar["g"] & D               
\end{tikzcd}
\end{equation*}
is a \emph{pullback square} if for every competing square $f\cdot p'=g\cdot q'$, there exists a unique morphism~$k$ such that $p\cdot k=p'$ and $q\cdot k= q'$. The notion of \emph{weak pullback square} is defined in the same way except that~$k$ is not required to be unique. A functor \emph{preserves weak pullbacks} if it maps weak pullback squares to weak pullback squares, and preserves \emph{inverse images} if it maps pullback squares containing an injective map in the cospan to pullback squares. Since preserving weak pullbacks entails preserving injections, it follows that preserving weak pullbacks is stronger than preserving inverse images.

\begin{example}
  \label{p:30}
  The following functors preserve weak pullbacks.
  \begin{enumerate}
  \item Every polynomial functor.
  \item The powerset functor~$\ftP$.
  \item The finite powerset functor~$\pfin$.
   \item The filter functor~$\calF$.
  \item The ultrafilter functor~$\ultra$.
  \item The discrete distribution functor~$\dist$.
  \item The monoid-valued functor $\NN^{(-)}$.     
  \end{enumerate}
\end{example}
\begin{example}
  \label{p:31}
  The following functors do \emph{not} preserve weak pullbacks.
  \begin{enumerate}
  \item The functor $(-)^3_2$ given on sets~$X$ by
    $(X)^3_2=\{(x,y,z)\in X^3\mid |\{x,y,z\}|\le 2\}$.
  \item The restricted powerset functor $\ftP_n$ for $1 < n<\omega$.
  \item The monotone neighbourhood functor~$\mon$ and its finitary
    part~$\mon_\omega$.
  \item The neighbourhood functor~$\neighb$ and its finitary
    part~$\neighb_\omega$.
  \item The monoid-valued functor $\ZZ^{(-)}$.
  \end{enumerate}
\end{example}
To conclude this section, we present a few examples of Barr
extensions.
\begin{example}
  Let~$r\colon X\relto Y$. 
  \begin{enumerate}
  \item Let~$F=(-)^2$ be the squaring functor, and let $\ftbF$ be the
    Barr extension of~$\ftF$. Then $(x,y)\mathrel{\ftbF r}(x',y')$ iff
    $x\mathrel{r}x'$ and $y\mathrel{r}y'$.
  \item The Barr extension~$\overline\ftP$ of the powerset
    functor~$\ftP$ is the Egli-Milner extension; that is,
    $A\mathrel{\overline P r}B$ iff for every $x\in A$, there is
    $y\in B$ such that $x\mathrel{r}y$, and symmetrically.
  \item The Barr extension~$\overline\dist$ of the discrete
    distribution functor~$\dist$ is described as follows: Two
    distributions $\mu\in\dist X$, $\nu\in\dist Y$ are related under
    $\overline\dist r$ iff there exists a \emph{coupling}
    $\rho\in\dist r$ of~$\mu,\nu$, i.e.~a distribution~$\rho$ on~$r$
    (as a set) whose marginals are~$\mu$, $\nu$ in the sense that
    $\dist\pi_1(\rho)=\mu$, $\dist\pi_2(\rho)=\nu$ for~$\pi_1,\pi_2$
    as above.
  \end{enumerate}
\end{example}
\section{\dots \ It Is the Barr Extension}
We proceed to introduce a class of functors characterized by a
property that we term \emph{elementwise boundedness}. This class
includes all accessible functors (also known as bounded functors). We
show that if such a functor admits an extension, then this extension
is the Barr extension.

In the remainder of this section, we fix a functor
$\ftF \colon \SET \to \SET$ and an extension
$\extF \colon \REL \to \REL$ of $\ftF$.  We record some preliminary
observations.

\begin{lemma}
        \label{p:1}
        The following statements hold.
        \begin{enumerate}
                \item
                        \label{p:11}
                        For every injective function \(i \colon X \to Y\),
                        \((\ftF i)^\circ \leq \extF  (i^\circ)\).
                \item
                        \label{p:2}
                        For every surjective function \(e \colon X \to Y\),
                        \(\extF (e^\circ) \leq (\ftF e)^\circ\).
                \item
                        \label{p:24}
                        $\extF$ is completely determined by its action on
                        converses of injections and surjections.
        \end{enumerate}
\end{lemma}
\begin{proof}
        \quad
        \begin{enumerate}
                \item
                        Let \(i \colon X \to Y\) be an injective function.
                        Then \(i^\circ \cdot i = 1_X\).
                        Therefore, 
                        \(
                                (\ftF i)^\circ
                                = \extF (i^\circ \cdot i) \cdot (\ftF i)^\circ
                                = \extF  (i^\circ) \cdot \ftF i \cdot (\ftF i)^\circ \leq \extF  (i^\circ)
                        \).
                \item
                        Let \(e \colon X \to Y\) be a surjective function.
                        Then \(e \cdot e^\circ = 1_Y\).
                        Therefore, 
                        \(
                                \extF (e^\circ) \leq (\ftF e)^\circ \cdot \ftF e \cdot \extF  (e^\circ)
                                = (\ftF e)^\circ \cdot \extF (e \cdot e^\circ)
                                = (\ftF e)^\circ
                        \).
                \item
                        Let $r \colon X \relto Y$ be a relation with canonical factorization $\rho \cdot \pi^\circ$, and let $\pi = i \cdot e$ be the image factorization of $\pi$.
                        Then,
                        \(
                                \extF(\rho \cdot e^\circ \cdot i^\circ)
                                = \ftF \rho \cdot \extF(e^\circ) \cdot \extF(i^\circ).
                        \)
                        Therefore,~$\extF r$ is completely determined by $\extF (i^\circ)$ and $\extF(e^\circ)$, and~$e$ is surjective and~$i$ is injective.
                        \qed
        \end{enumerate}
        \noqed
\end{proof}
We now turn to locally monotone extensions.
In the sequel, given a set $X$,
let $\nabla_X \colon X + X \to X$
denote the codiagonal, and
let $\rho_1 \colon X \to X + X$ and
$\rho_2 \colon X \to X +X$ denote
the injections into the first and second components, respectively.

In the next result, we will see that local monotonicity can be detected either by inspecting the action of an extension on converses of injections, or on converses of surjections.

\begin{lemma}
        \label{p:6}
        The following statements are equivalent:
        \begin{enumerate}
                \item
                        \label{p:8}
                        The extension \(\extF \) is locally monotone.
                \item
                        \label{p:9}
                        For every relation \(r \colon X \relto X\),
                        \(
                                r \leq 1_X \Rightarrow \extF  r \leq 1_{\ftF X}
                        \).
                \item
                        \label{p:4}
                        For every relation \(r \colon X \relto X\),
                        \(
                                r \geq 1_X \Rightarrow \extF  r \geq 1_{\ftF X}
                        \).
                \item
                        \label{p:10}
                        For every injective function \(i \colon X \to Y\),
                        \(
                                \extF  (i^\circ) = (\ftF i)^\circ
                        \).
                \item
                        \label{p:5}
                        For every surjective function \(e \colon X \to Y\),
                        \(\extF (e^\circ) = (\ftF e)^\circ\).
                \item
                        \label{p:23}
                        For every set $X$,
                        \(\ftF \rho_1 \leq \extF (\nabla_X^\circ)\).
        \end{enumerate}
\end{lemma}
\begin{proof}
        We begin by showing
        \ref{p:8}
        \(\Rightarrow\) \ref{p:9}
        \(\Rightarrow\) \ref{p:10}
        \(\Rightarrow\) \ref{p:8},
        then we will show
        \ref{p:10}
        \(\Rightarrow\) \ref{p:23}
        \(\Rightarrow\) \ref{p:9},
        and
        \ref{p:8}
        \(\Rightarrow\) \ref{p:4}
        \(\Rightarrow\) \ref{p:5}
        \(\Rightarrow\) \ref{p:23}.
                \begin{itemize}[wide]
                \item[\ref{p:8} \(\Rightarrow\) \ref{p:9}] 
                        Trivial.
                \item[\ref{p:9} \(\Rightarrow\) \ref{p:10}]
                        Let \(i \colon X \to Y\) be an injective function.
                        From \Cref{p:1} we have 
                        \((\ftF i)^\circ \leq\extF (i^\circ)\).
                        To see that the converse holds,
                        note that \(i \cdot i^\circ \leq 1_X\).
                        Hence, 
                        \(
                                \ftF i \cdot\extF  (i^\circ)
                                = \extF  (i \cdot i^\circ) \leq \extF  1_X
                                = 1_{\ftF X}
                        \).
                        Therefore, 
                        \(
                                \extF  (i^\circ)
                                \leq (\ftF i)^\circ \cdot  \ftF i \cdot\extF  (i^\circ)
                                \leq (\ftF i)^\circ
                        \).
                \item[\ref{p:10} \(\Rightarrow\) \ref{p:8}]
                        Let \(r,s \colon X \relto Y\) be relations such that \(r \leq s\), and
                        let $X \xleftarrow{\pi_r} R \xrightarrow{\rho_r} Y$ and
                        $X \xleftarrow{\pi_s} S \xrightarrow{\rho_s} Y$ be the
                        spans corresponding to the canonical factorizations of $r$ and $s$,
                        respectively.
                        Then, as \(r \leq s\),
                        \(\pi_r = \pi_s \cdot i\) and  
                        \(\rho_r = \rho_s \cdot i\), 
                        where \(i \colon R \hookrightarrow S\) denotes the inclusion map.
                        Hence,
                        \(
                                \extF  r
                                = \extF  (\rho_r \cdot {\pi_r}^\circ)
                                = \ftF \rho_r \cdot \extF ({\pi_r}^\circ)
                                = \ftF \rho_s \cdot \ftF i \cdot \extF (i^\circ) \cdot \extF ({\pi_s}^\circ)
                        \).
                        Therefore, as 
                        \(\extF (i^\circ)  = (\ftF i)^\circ\) and 
                        \(\ftF i \cdot (\ftF i)^\circ \leq 1_S\),
                        \(\extF  r \leq \ftF \rho_s \cdot \extF  ({\pi_s}^\circ) = \extF  s\).
                \item[\ref{p:10} \(\Rightarrow\) \ref{p:23}]
                        Let $X$ be a set.
                        By definition, $1_X = \nabla_X \cdot \rho_1 $,
                        which entails
                        $ 1_X = \rho_1^\circ \cdot \nabla_X^\circ$.
                        Hence,
                        \(
                                \ftF \rho_1
                                = \ftF \rho_1 \cdot \extF(\rho_1^\circ \cdot \nabla_X^\circ)
                                = \ftF \rho_1 \cdot \extF(\rho_1^\circ) \cdot \extF(\nabla_X^\circ)
                                = \ftF \rho_1 \cdot (\ftF \rho_1)^\circ \cdot \extF(\nabla_X^\circ)
                                \leq \extF(\nabla_X^\circ),
                        \)
                        where the last equality holds by hypothesis since $\rho_1$ is injective.
                \item [\ref{p:23} \(\Rightarrow\) \ref{p:9}]
                         Let \(r \colon X \relto X\) be a relation such that \(r \leq 1_X\).
                         Consider the function
                         $\phi \colon X \to X + X$ that sends $x \in X$ to $\rho_1(x)$,
                         if $x\, r\, x$, and to $\rho_2(x)$, otherwise.
                         Then, \(\phi^\circ \cdot \rho_1 = r\) and $\nabla_X \cdot \phi = 1_X$.
                        Hence, 
                        \(\extF (\phi^\circ) \cdot\ftF \rho_1 = \extF  r\) and
                        \(\extF  (\phi^\circ) \cdot\extF (\nabla_X^\circ) = 1_{\ftF X}\).
                        Therefore, by hypothesis, $\extF  r \leq 1_{\ftF X}$.
                \item[\ref{p:8} \(\Rightarrow\) \ref{p:4}]
                        Trivial.
                \item[\ref{p:4} \(\Rightarrow\) \ref{p:5}]
                        Let \(e \colon X \to Y\) be a surjective function.
                        From \Cref{p:1} we have
                        \(\extF (e^\circ) \leq (\ftF e)^\circ\).
                        To see that the reverse inequality holds,
                        note that \(e^\circ \cdot e \geq 1_X\).
                        Hence,
                        \(
                                \extF (e^\circ) \cdot\ftF e
                                = \extF  (e^\circ \cdot e)
                                \geq \extF  1_X = 1_{\ftF X}
                        \).
                        Therefore,
                        \(
                                \extF (e^\circ)
                                \geq \extF (e^\circ) \cdot\ftF e \cdot (\ftF e)^\circ
                                \geq (\ftF e)^\circ
                        \).
                \item[\ref{p:5} \(\Rightarrow\) \ref{p:23}]
                        Let $X$ be a set.
                        Then, as $\nabla_X$ is surjective and $\nabla_X \cdot \rho_1 = 1_X$,
                        by hypothesis,
                        \(
                                (\ftF \rho_1)^\circ \cdot \extF (\nabla_X^\circ)
                                = (\ftF \rho_1)^\circ \cdot (\ftF \nabla_X)^\circ
                                = 1_{\ftF X}
                        \).
                        Hence, by adjointness, $\ftF \rho_1 \leq \extF (\nabla_X^\circ)$.
                        \qed
        \end{itemize}
        \noqed
\end{proof}
As a consequence, we obtain the following result due to Carboni et al \cite{CKW91} (see also \cite{KV16}).

\begin{theorem}
        \label{p:12}
        A functor admits a locally monotone extension if and only if it preserves weak pullbacks, and in this case such an extension is unique, namely its  Barr extension.
\end{theorem}
\begin{proof}
        In \Cref{p:1} we have seen that every extension is completely determined by its action on converses of injective and surjective maps, while in \Cref{p:6} we have seen that every locally monotone extension agrees on such relations with the Barr extension, which is exclusive to functors that preserve weak pullbacks.
\qed\end{proof}
To construct multiple extensions for a weak pullback-preserving
functor, one can think of modifying its Barr extension slightly.
Indeed, in a similar context, it has been shown that it is possible to
slightly enlarge  the Barr extension of an exponential functor to
obtain other \emph{normal lax extensions} \cite{GHN+25b}.
 However, this fails for (strict) extensions,  since there is no extension comparable with the Barr extension (other than itself) w.r.t. the pointwise order.
 
\begin{corollary}
        \label{p:50}
        Suppose that $\ftF$ preserves weak pullbacks.
        If $\ftbF \leq \extF $ or $\ftbF \geq \extF $, then $\ftbF = \extF $.
\end{corollary}
\begin{proof}
        Suppose that $\ftbF \leq \extF$.
        Then, by \Cref{p:1}, for every surjection $e$,
        $(\ftF e)^\circ = \ftbF(e^\circ) \leq \extF(e^\circ) \leq (\ftF e)^\circ$.
        Therefore, by \Cref{p:6}, $\extF$ is locally monotone, from which we conclude that it coincides with the Barr extension due to \Cref{p:12}.
        The other case follows analogously.
\qed\end{proof}
This suggests that it may be difficult to find examples of functors that admit multiple extensions.
 In fact, we will see that the large class of functors introduced next admits at most one extension.
 We first fix some notation.
    
 We use calligraphic letters, such as $\calK$ and $\calJ$,
 to denote $\SET$-indexed families of sets, with the exception being constant families which may be represented by the set that determines them, and
we write $\calK \leq \calJ$ if $\calK$ is included componentwise in $\calJ$.
Furthermore, given a subset $A$ of $X$, and an element $a \in \ftF X$,
we write $\fx \in \ftF A$ to denote that $a \in \ftF i_A [\ftF A]$,
where $i_A \colon A \rightarrowtail X$ is the corresponding inclusion.
The set $\{0,1\}$ is denoted by $2$.

\begin{defn}
  Let $\calK$ be a $\SET$-indexed family of sets such that
  $2 \leq \calK$.  The functor $\ftF$ is \emph{elementwise
    $\calK$-bounded} if for every set $X$ and every
  \(\fx \in \ftF (X \times \calK_X)\) there is a set
  \(A \subseteq X \times \calK_X\) such that \(\fx \in \ftF A\)
  and~$A$ does not contain \(\{x\} \times \calK_X\) for any
  \(x \in X\).  The functor $\ftF$ is \emph{elementwise bounded} if it
  is elementwise $\calK$-bounded for some family $\calK$.
\end{defn}
\noindent 
Intuitively, we regard $X \times \calK_X$ as providing $\calK_X$-many
\emph{copies} of every element $x\in X$. Then, $\ftF$ is elementwise
$\calK$-bounded if given a set $X$, every element of
$\ftF(X \times \calK_X)$ can be constructed without ever needing to
use all copies of any element of~$X$.

\begin{remark}\label{rem:trivial}
  The condition defining elementwise boundedness is vacuously
  satisfied on the empty set.
\end{remark}

\begin{example}
        \label{p:32}
        The following functors are elementwise bounded.
        \begin{enumerate}
                \item
                        Every constant functor $S$ is elementwise  $2$-bounded since,
                        for every set $X$,
                        each element of $S(X \times 2)$ belongs to $\ftF\varnothing$.
                \item
                        Every exponentiable functor $(-)^S$ for a non-empty exponent is elementwise ($S+1$)-bounded since,
                        for every set $X$,
                        each function from $S$ to $X \times (S+1)$   selects at most~$|S|$ copies of each element of $X$.
                \item
                        Every functor ``sum with a fixed set $S$'' $(- + S)$  is elementwise $2$-bounded since, for every set $X$, each element of $X \times 2 +S$ belongs to the image under $(-)+S$ of the empty set or of a singleton.
                \item
                        Every functor that preserves binary coproducts is elementwise  $2$-bounded.
                        To see this, note that given a set $X$, as \(X \times 2 \simeq X + X\), if \(\ftF\) preserves coproducts, every element of \( \ftF(X \times 2)\) either belongs to
        \(\ftF(X \times \{0\})\) or to \(\ftF(X \times \{1\})\).
        
                        In particular, every functor $A\times(-)$ ``product with a fixed set'' is elementwise $2$-bounded, and  
                        the ultrafilter functor, which is \emph{not} accessible,
                        is elementwise $2$-bounded.
                \item
                        Every $\lambda$-accessible functor is elementwise $\lambda$-bounded.
                        Indeed, suppose that~$\ftF$ is $\lambda$-accessible.
                        Let $X$ be a set.
                        Then, for every element \(\fx \in \ftF (X \times \lambda)\)
                        there is \(A \subseteq X \times \lambda\) with \(|A| < \lambda\)
                        such that \(\fx \in \ftF A\).
                        Moreover, as \(|A| < \lambda\), for every \(x \in X\),
                        \(\{x\} \times \lambda \not\subseteq A\).
                        
                        In particular, all functors of \Cref{p:28} are elementwise bounded.
        \end{enumerate}
\end{example}
\Cref{p:32} suggests that most functors used for coalgebraic modelling are elementwise $\kappa$-bounded for some set~$\kappa$. In many cases, elementwise boundedness can be realized by a fixed cardinal:

\begin{proposition}
        Consider the following statements.
        \begin{enumerate}
                \item
                        \label{p:41}
                        $\ftF$ is elementwise bounded.
                \item
                        \label{p:42}
                        There is a set $\kappa$ of cardinality greater than one with the property that for every $\fx \in \ftF \kappa$ there is $A \subsetneq \kappa$ such that $\fx \in \ftF A$.
        \end{enumerate}
        Then, \ref{p:41} $\Rightarrow$ \ref{p:42}, and when $\ftF$ preserves binary products \emph{or} inverse images, \ref{p:42} $\Rightarrow$ \ref{p:41}, in which case $\ftF$ is elementwise $\kappa$-bounded.
\end{proposition}
\begin{proof}
        The first claim is immediate from the definition of elementwise boundedness, applied to  $X = 1$.
        To see the second claim, suppose that $\ftF$ preserves inverse images.
        Let $X$ be a set, and
        let $\fx$ be an element of $\ftF (X \times \kappa)$.
        Then, by hypothesis, there is a set $A \subsetneq \kappa$ such that $\ftF \pi_2(a) \in \ftF A$, where $\pi_2 \colon X \times \kappa \to \kappa $ is the corresponding projection,
        This means that $\fx$ belongs to the pullback of $\ftF(X \times \kappa) \xrightarrow{\ftF \pi_2} \ftF \kappa \xleftarrow{\ftF i_A} \ftF A$, where $i_A \colon A \rightarrowtail X$ denotes the corresponding inclusion,
        which is given by $\ftF(\pi_2^{-1}[A])$, because $\ftF$ preserves inverse images.
        Therefore, $\fa \in \ftF(\pi_2^{-1}[A])$, and since $A \subsetneq \kappa$, for every $x \in X$, $\{x\} \times \kappa \not\subseteq \pi_2^{-1}[A]$.
        The case of functors that preserve binary products is trivial.
\qed
\end{proof}
        
\begin{example}
        The following functors fail to be elementwise bounded.
        \begin{enumerate}
                \item The powerset functor $\ftP$, since there is no set $\kappa$ that belongs to the image under $\ftP$ of any of its proper subsets.
                \item The monotone neighbourhood functor $\ftM$, since there is no set $\kappa$ for which the neighbourhood system $\{\kappa\}$ belongs to the image under $\ftM$ of any of its proper subsets.
                \item The neighbourhood functor, by the same reasoning as for the monotone neighbourhood functor.
        \end{enumerate}
\end{example}
Finally, we are ready to show that every functor
that is elementwise bounded admits only locally monotone extensions,
which, as we have seen in \Cref{p:12},
are exclusive to functors that preserve weak pullbacks.

\begin{lemma}
        \label{p:40}
        For every map $e \colon Y \to X$,
        \[
                \ftF e \geq \bigvee \{(\ftF \rho_1)^\circ \cdot \ftF f \mid e \cdot f^\circ = \nabla_X\},
        \]
        where $f$ ranges over all functions from $Y$ to $X + X$.
\end{lemma}
\begin{proof}
        Let $f \colon Y \to X + X$ be a function such that
        $e\cdot f^\circ = \nabla_X$.
        Then, given that, independently of $\ftF$ preserving weak pullbacks,
        the Barr construction is well-defined \cite{Bar70},
        $\ftF e \cdot (\ftF f)^\circ = \ftF \nabla_X$.
        Hence,
        \[
                \ftF e \cdot \big((\ftF \rho_1)^\circ \cdot  \ftF f\big)^\circ
                = \ftF e \cdot (\ftF f)^\circ \cdot \ftF \rho_1
                = \ftF(\nabla_X \cdot \rho_1)
                = 1_{\ftF X}.
        \]
        Therefore, $(\fy,\fx) \in (\ftF\rho_1)^\circ \cdot \ftF f$ entails $\ftF e (\fy) = \fx$.
\qed\end{proof}
The next result is the key ingredient in the proof of our main theorem.
It shows that the action on surjections of functors that are elementwise bounded can be expressed in terms of factorizations of codiagonal maps,
which detect locally monotone extensions,
as we have seen in \Cref{p:6}.

\begin{lemma}
        \label{p:37}
        Consider the following statements.
        \begin{enumerate}
                \item
                        \label{p:35}
                        The functor $\ftF$ is elementwise $\calK$-bounded.
                \item
                        \label{p:36}
                        For every set $X$,
                        there are a set~$Y_X$ and a surjection $e \colon Y_X \twoheadrightarrow X$ such that
                        \begin{equation}
                                \label{p:39}
                                \ftF e
                                = \bigvee\{(\ftF \rho_1)^\circ \cdot \ftF f \mid e \cdot f^\circ = \nabla_X\}
                        \end{equation}
                        where $f$ ranges over all functions from $Y_X$ to $X + X$.
        \end{enumerate}
        Then, \ref{p:35} $\Rightarrow$ \ref{p:36}, and
        when $\ftF$ preserves inverse images,
        \ref{p:36} $\Rightarrow$ \ref{p:35}.
      \end{lemma}
      \noindent (Recall that~$\rho_1$ is the first injection $X\hookrightarrow X+X$.)
\begin{proof}
        \ref{p:35} $\Rightarrow$ \ref{p:36}.
        It suffices to show the reverse inequality of \Cref{p:40}.
        To this end,
        let $X$ be a set, and
        let $\pi_1 \colon X \times \calK_X \to X$ be the projection to $X$.
        As $\calK_X$ is non-empty by definition, $\pi_1$ is surjective, and
        we claim that it satisfies \eqref{p:39}.
        To see this, let $\fa \in \ftF(X \times \calK_X)$.
        By hypothesis, there is a set $A \subseteq X \times \calK_X$
        such that $\fx \in \ftF A$, and
        for every $x \in X$, $\{x\} \times \calK_X$ is not contained in $A$.
        Now, consider the function $f \colon X \times \calK_X \to X + X$ defined by
        $f(x,y) = \rho_1(x)$, if $(x,y) \in A$, and
        $f(x,y) = \rho_2(x)$, otherwise.
        Since $2 \leq |\calK_X|$, and $\{x\} \times \calK_X \not\subseteq A$,
        for every $x \in X$,
        we can assume without loss of generality that
        for every $x \in X$ there are $k_1,k_2 \in \calK_X$ such that
        $(x,k_1) \in A$ and $(x,k_2) \notin A$.
                                Hence, $\pi_1 \cdot f^\circ \geq \nabla_X$, and
        as the other inequality holds by definition of $f$,
        we obtain $\pi_1 \cdot f^\circ = \nabla_X$.
        Furthermore, $\rho_1 \cdot \pi_1 \cdot i_A = f \cdot i_A $,
        where $i_A \colon A \rightarrowtail X \times \calK_X$ is the corresponding inclusion. 
        This entails $\ftF \rho_1 \cdot \ftF \pi_1 \cdot \ftF i_A = \ftF f \cdot \ftF i_A$.
        Thus, $\ftF \pi_1 \cdot \ftF i_A \leq (\ftF \rho_1)^\circ \cdot \ftF f \cdot \ftF i_A$.
        Therefore, as $\fx \in \ftF i_A[\ftF A]$ by hypothesis,
        we conclude that $(\fx,\ftF \pi_1(\fx)) \in (\ftF \rho_1)^\circ \cdot \ftF f$.
                
        \ref{p:36} $\Rightarrow$ \ref{p:35}.
        Let $X$ be a set, and
        let $e \colon Y_X \twoheadrightarrow X$ be a surjective map that satisfies \eqref{p:39}.
        First we note that for every surjection $e' \colon Y \twoheadrightarrow Y_X$,
        the surjection $e \cdot e'$ also satisfies \eqref{p:39}.
        Let $e' \colon Y \twoheadrightarrow Y_X$ be a surjective map.
        Then,
        \begin{align*}
                \ftF(e \cdot e')
                &= \bigvee \{(\ftF \rho_1)^\circ \cdot \ftF f \cdot \ftF e' \mid e \cdot f^\circ = \nabla_X\} \\
                &= \bigvee \{(\ftF \rho_1)^\circ \cdot \ftF (f  \cdot e') \mid e \cdot e' \cdot (f \cdot e')^\circ = \nabla_X\}
        \end{align*}
        where $f$ ranges over all functions from $Y_X$ to $X + X$.
        Hence, due to \Cref{p:40},
        \begin{displaymath}
                \ftF (e \cdot e')
                = \bigvee\{(\ftF \rho_1)^\circ \cdot \ftF g \mid e \cdot e' \cdot g^\circ = \nabla_X\},
        \end{displaymath}
        where $g$ ranges over all functions from $Y$ to $X + X$.
        Now, if $X$ is empty, then by \Cref{rem:trivial}. we can simply take $\calK_X = 2$.
                                Hence, suppose that $X$ in non-empty.
        Then, as $e$ is a surjection, $Y_X$ needs to be non-empty,
        and thus, due to our argument on the composite of surjections,
        we can assume without loss of generality that $2 \leq |Y_X|$.
        Next, we choose for each \(x \in X\) a surjective map $e'_x$ from \(Y_X\) to \(e^{-1}(x)\),
        which exists because $e^{-1}(x) \subseteq Y_X$, by definition.
        This yields a  map \(e' \colon X \times Y_X \to Y_X\) that sends $(x,y) \in X \times Y_X$ to $e'_x(y)$,
        which is surjective because $Y = \bigcup_{x \in X} e^{-1}(x)$. 
        Furthermore,  by construction,
        $e \cdot e'$ coincides with the projection $\pi_1 \colon X \times Y_X \to X$.
        Hence, due to our argument on the composite of surjections,
        we conclude that $\pi_1 \colon X \times Y_X \to X$ satisfies \eqref{p:39}.
        Now, let $\fx \in \ftF(X \times Y_X)$.
        Then, there is $f \colon X \times Y_X \to X + X$ such that
        $\pi_1 \cdot f^\circ = \nabla_X$ and
        $(\fx, \ftF\pi_1(\fx)) \in (\ftF \rho_1)^\circ \cdot \ftF f$,
        which means $\ftF f(\fx) = \ftF \rho_1 \cdot \ftF \pi_1(\fx)$.
        Consider the set $A = \{(x,y) \in X \times Y_X \mid f(x,y) \in \rho_1[X]\}$.
        We claim that for every $x \in X$, $\{x\} \times Y_X$ is not contained in $A$,
        and that $\fx\in \ftF A$.
        Let $x \in X$.
        Then, as $\pi_1 \cdot f^\circ = \nabla_X$,
        there is $y \in Y_X$ such that $f(x,y) = \rho_2(x)$.
        Hence, $(x,y) \notin A$.
        On the other hand,
        note that, as $\ftF f(\fx) = \ftF \rho_1 \cdot \ftF \pi_1(\fx)$,
        $(\fx, \ftF\pi_1(\fx))$ belongs to the pullback of
        $\ftF Y_X \xrightarrow{\ftF f} \ftF(X + X) \xleftarrow{\ftF\rho_1} \ftF X$.
        Hence, as $\rho_1$ is injective, $A$ is the pullback of $Y_X \xrightarrow{f} X + X \xleftarrow{\rho_1} X$,
        and $\ftF$ preserves inverse images, we conclude that $\fx \in \ftF A$.
        Therefore, $\ftF$ is elementwise $\calK$-bounded,
        where $\calK_\varnothing = 2$,
        and for every other set $X$, $\calK_X = Y_X$.
\qed\end{proof}
Now, our main result follows straightforwardly.

\begin{theorem}
        \label{p:13}
        A functor that is elementwise bounded admits an extension if and only if it preserves weak pullbacks, and in this case such an extension is unique, namely its  Barr extension.
\end{theorem}
\begin{proof}
        Suppose that $\ftF$ is elementwise $\calK$-bounded.
        By \Cref{p:12} it suffices to show that $\extF$ is locally monotone. 
        Let \(X\) be a set.
        We claim that \(\ftF \rho_1 \leq \extF (\nabla_X^\circ)\),
        from which we conclude that \(\extF \) is locally monotone by \Cref{p:6}.
        First we note that, by \Cref{p:37}, there is a surjection $e \colon Y_X \twoheadrightarrow X$ such that
        \begin{equation}
                \label{p:38}
                \ftF e
                = \bigvee\{(\ftF \rho_1)^\circ \cdot \ftF f \mid e \cdot f^\circ = \nabla_X\},
        \end{equation}
        where $f$ ranges over all functions from $Y_X$ to $X + X$.
        Now, let $\fx \in \ftF X$.
        Since $e$ is a surjection, $e \cdot e^\circ = 1_X$, which entails
        $\ftF e \cdot \extF(e^\circ) = \extF(e \cdot e^\circ) = 1_{\ftF X}$. 
        Hence, there is $\fy \in \ftF Y$ such that
        $(\fx,\fy) \in \extF(e^\circ)$ and $\ftF e(\fy) = \fx$.
        Then, by \eqref{p:38}, there is $f \colon Y \to X + X$ such that
        $(\fy,\fx) \in (\ftF \rho_1)^\circ \cdot \ftF f$,
        which means $\ftF f(\fy) = \ftF \rho_1(\fx)$, and
        $f \cdot e^\circ = \nabla_X^\circ$.
        Hence,
        $\ftF f \cdot \extF(e^\circ) = \extF(f \cdot e^\circ) = \extF(\nabla_X^\circ)$,
        $(\fx,\fy) \in \extF(e^\circ)$, and $\ftF f (\fy) = \ftF \rho_1(\fx)$.
        Therefore, $(\fx, \ftF\rho_1(\fx)) \in \extF(\nabla_X^\circ)$.
\qed\end{proof}
As an immediate consequence, we obtain that every accessible functor admits at most one extension,  equivalently at most one Kleisli distributive law over the powerset monad.

\begin{corollary}
        A functor that is accessible \emph{or} preserves binary coproducts admits an extension if and only if it preserves weak pullbacks, and in this case such an extension is unique, namely its  Barr extension.
\end{corollary}
This result can be specialized further.
For instance, it is well-known that a monoid-valued functor preserves weak pullbacks
iff the monoid is positive and refinable~\cite{GS01}.
Therefore:

\begin{corollary}
        A monoid-valued functor admits an extension if and only if the monoid is positive and refinable, and in this case such an extension is unique, namely its  Barr extension.
\end{corollary}
As preservation of weak pullbacks is a well-studied  property, our results make it is easy to identify functors that admit a unique extension, or that do not admit any extension at all.
For instance, all functors of \Cref{p:30} admit a unique extension, while none of the finitary functors of \Cref{p:31} admits an extension.

\section{\dots \ And Sometimes More}
We have shown that for elementwise bounded functors, such as the finite powerset functor, weak pullback-preservation is a necessary condition for the existence of extensions and a sufficient condition for their uniqueness.
In this section, we show that this does not hold in general.
Notably, we show that, in contrast to the finite powerset functor, the full powerset functor admits exactly three extensions.

Our examples are based on the fact that monad morphisms induce Kleisli laws.
Recall that given monads $\mS = (\ftS,e^\ftS, m^\ftS)$ and $\mT = (\ftT, e^\ftT,m^\ftT)$,
a natural transformation $\lambda \colon \mS \to \mT$ is a monad morphism
if it makes the following diagrams commute:

\begin{displaymath}
\begin{array}{c c}
\begin{tikzcd}
        \ftId & \ftS  \\
          & \ftT
        \ar[from=1-1, to=1-2, "e^\ftS"]
        \ar[from=1-1, to=2-2, "e^\ftT"']
        \ar[from=1-2, to=2-2, "\lambda"]
\end{tikzcd}
&
\hspace{3em}
\begin{tikzcd}
        \ftS^2 & \ftT\ftS & \ftT^2 \\
        \ftS   & & \ftT.
        \ar[from=1-1,to=1-2,"\lambda\ftS"]
        \ar[from=1-1,to=2-1,"m^\ftS"']
        \ar[from=1-2,to=1-3,"\ftT\lambda"]
        \ar[from=1-3,to=2-3,"m^\ftT"]
        \ar[from=2-1,to=2-3,"\lambda"']
\end{tikzcd}
\end{array}
\end{displaymath}

\begin{proposition}
        Let $\mS=(\ftS,e^\ftS,m^\ftS)$ and
        $\mT=(\ftT,e^\ftT,m^\ftT)$ be monads.
        For every monad morphism $\lambda \colon \mS \to \mT$,
        the natural transformation $e^\ftS \ftT \cdot m^\ftT \cdot \ftT\lambda$
        yields a Kleisli law $\ftT\mS \to \mS\ftT$.
\end{proposition}
\begin{proof}
        Let $\lambda \colon \mS \to \mT$ be a monad morphism.
        Note that the axioms of monad, monad morphism, and natural transformation entail that the  diagrams below commute.
        \begin{displaymath}
        \begin{array}{c c}
        \begin{tikzcd}
                \ftT & \ftT\ftS  \\
                     & \ftT^2 \\
                     & \ftT \\
                     & \ftS\ftT
                \ar[from=1-1, to=1-2, "\ftT e^\ftS"]
                \ar[from=1-1, to=2-2, "\ftT e^\ftT", bend right]
                \ar[from=1-2, to=2-2, "\ftT \lambda"]
                \ar[from=1-1, to=4-2, "e^\ftS \ftT"', bend right]
                \ar[from=1-1, to=3-2, bend right, equal]
                \ar[from=2-2,to=3-2, "m^\ftT"]
                \ar[from=3-2,to=4-2, "e^\ftS\ftT"]
\end{tikzcd}
&
\hspace{2em}
\begin{tikzcd}
        \ftT\ftS^2   &            &         & \ftT\ftS   & \\
        \ftT^2 \ftS  & \ftT^3     &         & \ftT^2     & \\
        \ftT\ftS     & \ftT^2     &         & \ftT       & \\
        \ftS\ftT\ftS & \ftS\ftT^2 &\ftS\ftT & \ftS^2\ftT & \ftS\ftT
        \ar[from=1-1,to=1-4,"\ftT m^\ftS"]
        \ar[from=1-1,to=2-1,"\ftT\lambda\ftS"']
        \ar[from=2-1,to=3-1,"m^\ftT \ftS"']
        \ar[from=3-1,to=4-1,"e^\ftS\ftT\ftS"']
        \ar[from=1-4,to=2-4,"\ftT\lambda"]
        \ar[from=2-4,to=3-4,"m^\ftT"]
        \ar[from=3-4,to=4-5,"e^\ftS \ftT"]
        \ar[from=4-1,to=4-2,"\ftS\ftT\lambda"']
        \ar[from=4-2,to=4-3,"\ftS m^\ftT"']
        \ar[from=4-3,to=4-4,"\ftS e^\ftS\ftT"]
        \ar[from=4-4,to=4-5,"m^\ftS \ftT"]
        \ar[from=3-4,to=4-3,"e^\ftS\ftT"']
        \ar[from=3-1,to=3-2,"\ftT\lambda"]
        \ar[from=3-2,to=4-2,"e^\ftS\ftT^2"]
        \ar[from=2-1,to=2-2,"\ftT^2\lambda"]
        \ar[from=2-2,to=3-2,"m^\ftT \ftT"]
        \ar[from=2-2,to=2-4,"\ftT m^\ftT"]
        \ar[from=3-2,to=3-4,"m^\ftT"']
        \ar[from=4-3,to=4-5, equal, bend right]
\end{tikzcd}
\end{array}
\end{displaymath}
\qed\end{proof}
In particular,
a monad morphism \(\lambda \colon \mP \to \mT\)
from the powerset monad \(\mP = (\ftP, \eta, \mu)\)
to a monad \(\mT = (\ftT,e,m)\) gives rise, via the isomorphism \( \REL \simeq \SET_\mP\), 
to an extension of \(\ftT\) to \(\REL\) that
sends a relation \(r \colon X \relto Y\)
to (the graph of) the function \(m_Y \cdot\ftT(\lambda_Y \cdot r^\sharp)\),
where \(r^\sharp \colon X \to \ftP Y\) denotes the function corresponding
to the relation \(r\) under the isomorphism \(\REL \simeq \SET_\mP\).

\begin{example}
        \quad
        \label{p:22}
        \begin{enumerate}
                \item
                        For the monad morphism $\mP \to \mZ$ into the terminal monad
                        we obtain the Barr extension of the constant functor $1$.
                \item
                        For the identity monad morphism $\mP \to \mP$,
                        we obtain the \emph{relational image extension} that sends
                        a relation $r \colon X \relto Y$ to the function
                        that maps a set $A \subseteq X$ to its relational image $r[A]$.
                \item
                        \label{p:100}
                        For the monad morphism $\Box \colon \mP \to \mF$
                        (e.g. \cite[Examples~4.3]{SS08})
                         to the filter monad,
                        whose $X$-component is given by 
                        $\Box_X(U) = \{A \subseteq X \mid U \subseteq A\}$,
                        we obtain the extension that sends a relation $r \colon X \relto Y$
                        to the function
                        that maps a filter $\calA \in \calF X$
                        to $\upc\{r[A] \mid A \in \calA\}$ where $\upc$ denotes upwards closure w.r.t.~subset inclusion.        
                \item
                        \label{p:101}
                        For the monad morphism $\Box \colon \mP \to \mF \hookrightarrow \mM$
                        to the monotone neighbourhood monad,
                        we obtain an extension defined like the one for the filter functor.
                        Furthermore, for the monad morphism $\Diamond \colon \mP \to \mM$
                        (e.g. \cite[Examples~4.3]{SS08}),
                        whose $X$-component is given by
                        $\Diamond_X(U) = \{A \subseteq X \mid U \cap A \neq \varnothing\}$,
                        we obtain a different extension that
                        sends a relation $r \colon X \relto Y$ to the function
                        that maps a monotone neighbourhood system $\calA \in \ftM X$
                        to $\{B \mid r^\circ[B] \in \calA\}$.                        
        \end{enumerate}
\end{example}

By definition, an extension induced by a monad morphism sends each
relation to a function and, hence, generally behaves very differently
from the Barr extension (in particular, of course, such extensions
fail to be locally monotone).
\Cref{p:22} shows that there are functors that \emph{do not} preserve weak pullbacks,
such as the monotone neighbourhood functor (e.g. \cite[Example~3.10(2)]{GHN+25}),  yet admit (multiple) extensions, as well as functors that preserve weak pullbacks but admit multiple extensions,
such as the powerset functor and the filter functor.
In contrast, in the previous section we saw that the closely related  finite powerset functor and  ultrafilter functor admit a unique extension, meaning that elementwise boundedness is crucial to guarantee uniqueness in these cases.

In \Cref{p:50}, we have seen that it is \emph{not} possible to generate extensions by contracting or enlarging the Barr extension.
To conclude, we show that this may be possible for extensions induced by monad morphisms.
In fact, we show that the powerset functor has exactly three extensions:
the Barr extension, the relational image extension, and a slightly modified relational image extension, whose corresponding Kleisli distributive law sends $\{\varnothing\}$ to $\varnothing$ and, therefore, is not induced by a monad morphism.

\begin{proposition}\label{pro:3-laws}
There are exactly three extensions of the powerset functor:
\begin{enumerate}
  \item the Barr extension: $A \mathrel{(\extF r)} B$ iff $B \subseteq r[A]$ and $A \subseteq r^\circ[B]$;  
  \item the relational image extension: $A \mathrel{(\extF r)} B$ iff $B = r[A]$;
  \item the restricted relational image extension: $A \mathrel{(\extF r)} B$ iff either $B = r[A]\neq\emptyset$ or
  $A=B=\emptyset$.
  \end{enumerate}
\end{proposition}
\begin{proof}
We use the fact that extensions are equivalent to Kleisli distributive laws. So, 
let $\delta\c \ftF\ftP \to\ftP\ftF$ be such a law, where $\ftF$ is the powerset functor.
As a first step, we argue that $\delta$ must satisfy exactly one of the following 
conditions:
\begin{align}
\delta\{\{0,1\},\{2\}\} =&\; \{\{0,2\},\{1,2\},\{0,1,2\}\},\label{eq:barr-law}\\
\delta\{\{0,1\},\{2\}\} =&\; \{\{0,1,2\}\}\label{eq:em-law}
\end{align}
Let $f\c\{0,1,2\}\to\{1,2\}$ send $0,1$ to $1$ and $2$ to $2$, and let $p=\{\{0,1\},\{2\}\}$.
Then, using naturality of $\delta$ and the law $\delta(\ftF\eta) = \eta$, 
\begin{align*}
\{\{1,2\}\} =\delta \{\{1\},\{2\}\} = \delta (\ftF\ftP f)(p) = (\ftP\ftF f) (\delta (p)).
\end{align*}
Thus $\delta(p)$ must be such a set of subsets $s$ of $\{0,1,2\}$ that $(\ftP f)(s) = \{1,2\}$. 
This gives us the following cases:
\begin{enumerate}
  \item $\delta\{\{0,1\},\{2\}\}= \{\{0,2\}\}$,
  \item $\delta\{\{0,1\},\{2\}\}= \{\{1,2\}\}$,
  \item $\delta\{\{0,1\},\{2\}\}= \{\{0,2\},\{1,2\}\}$,
  \item $\delta\{\{0,1\},\{2\}\}= \{\{0,1,2\}\}$,
  \item $\delta\{\{0,1\},\{2\}\}= \{\{0,2\},\{0,1,2\}\}$,
  \item $\delta\{\{0,1\},\{2\}\}= \{\{1,2\},\{0,1,2\}\}$,
  \item $\delta\{\{0,1\},\{2\}\}= \{\{0,2\},\{1,2\},\{0,1,2\}\}$.
\end{enumerate}
Options 1, 2, 5 and 6 are incompatible with naturality of $\delta$: swapping 
$0$ and $1$ keeps the left-hand side unchanged, but not the right one. Let us also exclude~3. We have
\begin{align*}
\{\{0,2\},\{1,2\}\}
 =&\;\delta \{\{0,1\},\{2\},\{0,1\}\}\\*
 =&\;(\delta (\ftF\mu))\{\{\{0,1\},\{2\}\},\{\{0,1\}\}\}\\*
 =&\;(\mu (\ftP\delta)\delta)\{\{\{0,1\},\{2\}\},\{\{0,1\}\}\}\\
 =&\;(\mu (\ftP\delta))\{\{\{0,1\},\{0,1\}\},\{\{2\},\{0,1\}\}\}\\
 =&\;\mu\{\delta\{\{0,1\}\},\{\{2,0\},\{2,1\}\}\}\\
 =&\;\delta\{\{0,1\}\}\cup\{\{0,2\},\{1,2\}\},
\end{align*}
entailing $\delta\{\{0,1\}\}\subseteq\{\{0,2\},\{1,2\}\}$. This contradicts
naturality of $\delta$.
We are left exactly with the above options~\eqref{eq:barr-law} and~\eqref{eq:em-law}. 

Given a set $X$, and a non-empty family of non-empty mutually disjoint subsets ${(X_i\subseteq X)_{i\in I}}$,
assume
\begin{align}\label{eq:PP-dl}
\delta\{X_i\}_{i\in I} = \{Y_j\subseteq X\}_{j\in J}.
\end{align}
Let us show that for every $j\in J$ and every $i\in I$, $X_i\cap Y_j\neq\emptyset$. Fix some 
$i\in I$. If $X = X_i$, we are done straightforwardly. Assume that $X\neq X_i$,
 and let $f\c X\to\{0,1\}$ send the elements of $X_i$ to $1$ and the rest to 
$0$. Then 
\begin{align*}
\{\{0,1\}\} =\delta \{\{0\},\{1\}\} = \delta (\ftF\ftP f)\{X_i\}_{i\in I} = (\ftP\ftF f) (\delta \{X_i\}_{i\in I}).
\end{align*}
If $\delta \{X_i\}_{i\in I}$ contained a set disjoint from $X_i$, $\ftF f$ applied to 
this set would yield~$\{1\}$, contradiction.
Also, observe that by compatibility of $\delta$ with $\eta$,
\begin{align}\label{eq:d-e-ee}
\delta(\emptyset) = \{\emptyset\}.
\end{align}
Assume~\eqref{eq:barr-law}. Our goal is to derive the general formula for $\delta$:
\begin{align}\label{eq:barr-ne}
\delta\{X_i\}_{i\in I} = \Bigl\{Y\subseteq\bigcup_{i\in I} X_i\mid\forall i\in I.\, Y\cap X_i\neq\emptyset\Bigr\}.
\end{align}
We establish a special case first: In \eqref{eq:PP-dl}, take $I = I_1\cup I_2$, such that
$I_1\cap I_2=\emptyset$, $X_i = \{x_i,y_i\}$ if $i\in I_1$ and $X_i = \{x_i\}$
if $i\in I_2$. We proceed to prove that 
\begin{align}\label{eq:barr-21}
\{x_i\}_{i\in I}\in\delta(\{\{x_i,y_i\}\}_{i\in I_1}\cup\{\{x_i\}\}_{i\in I_2}).
\end{align}
If we can pick such $Y\in\delta\{X_i\}_{i\in I}$
that for every $i\in I_1$, either $x_i\notin Y$ or $y_i\notin Y$, we can prove~\eqref{eq:barr-21}
by applying naturality w.r.t.\ $f\c X\to X$ that swaps $x_i$ and~$y_i$ whenever $y_i\in Y$. 
Recall that for every $Y\in \delta\{X_i\}_{i\in I}$,
and every $i\in I_1$, $Y\cap\{x_i,y_i\} \neq\emptyset$, and thus $x_i\in Y$ or $y_i\in Y$.
Let us show that we can indeed pick such $Y$ that for every $i\in I_1$, $x_i$ and $y_i$
are not simultaneously in $Y$. Assume the opposite: for some $k\in I_1$, for all 
$Y\in\delta\{X_i\}_{i\in I}$, $\{x_k,y_k\}\subseteq Y$.
Let $f\c X\to\{0,1,2\}$ send $x_k$ to $0$, $y_k$ to $1$ and all the remaining elements 
to $2$. By naturality, every element of $\delta\{\{0,1\},\{2\}\}$ contains both $0$ and $1$,
contradicting~\eqref{eq:barr-law}.

Using~\eqref{eq:barr-21}, we continue with the proof of the general formula~\eqref{eq:barr-ne}:
As we argued above, every $Y\in\delta\{X_i\}_{i\in I}$ must have a non-empty intersection 
with every $X_i$. We are left to prove the opposite: for every $Y\subseteq X$,
if $Y\cap X_i\neq\emptyset$ for all $i\in I$, then $Y\in\delta\{X_i\}_{i\in I}$. 
Let us first prove a special case:
\begin{align*}
\bigcup_{i\in I} X_i\in\delta\{X_i\}_{i\in I}.
\end{align*}
Indeed, pick  $Z\in\delta\{\{x\}\times X_i\}_{i\in I,x\in X_i}$ (thus, $Z\in X\times\ftP(X)$),
and note that the elements of $Z$ have the form $(x,W)$, where $i\in I$, $x\in X_i$ 
and $W\neq\emptyset$. By applying  naturality w.r.t.\ the left projection $X\times\ftP(X)\to X$, we obtain that 
$\bigcup_{i\in I} X_i\in\delta\{X_i\}_{i\in I}$. Now, take an arbitrary
 $Y\subseteq X$ such that $Y\cap X_i\neq\emptyset$ for all $i\in I$ and partition $I$ disjointly into $I_1$ and $I_2$ so that $i\in I_2$
iff $X_i\subseteq Y$. Let $x_i = Y\cap X_i$, $y_i = X_i\setminus Y$ for $i\in I_1$
and $x_i = X_i$ for $i\in I_2$. Then
\begin{align*}
Y = \bigcup_{i\in I} x_i
 \in&\;\delta(\{x_i\}_{i\in I})\\*
 =&\;(\mu (\ftP\delta))\{\{x_i\}_{i\in I}\}\\*
 \subseteq&\;(\mu (\ftP\delta)\delta)\{\{x_i,y_i\}\}_{i\in I_1}\cup\{\{x_i\}\}_{i\in I_2})\\
 =&\;(\delta (\ftF\mu))(\{\{x_i,y_i\}\}_{i\in I_1}\cup\{\{x_i\}\}_{i\in I_2})\\*
 =&\;\delta\{X_i\}_{i\in I}.
\end{align*}
We have thus proved~\eqref{eq:barr-ne} for $I\neq\emptyset$ and non-empty $X_i$. By~\eqref{eq:d-e-ee},
the same formula remains valid for $I=\emptyset$. We are left to show that it also 
remains valid if $I\neq\emptyset$ and some of the~$X_i$ are empty. The following calculation, together with naturality of $\delta$, 
shows that $\delta\{\emptyset,\{1\}\} = \emptyset$: 
\begin{align*}
\delta\{\emptyset,\{1\}\}\subseteq&\;(\mu (\ftP\delta))\{\{\emptyset,\{1\}\},\{\{2\},\{1\}\},\{\emptyset,\{2\},\{1\}\}\}\\*
 =&\;(\mu (\ftP\delta)\delta)\{\{\emptyset,\{2\}\},\{\{1\}\}\}\\
 =&\;(\delta (\ftF\mu))\{\{\emptyset, \{2\} \},\{\{1\}\}\}\\*
 =&\;\delta\{\{2\},\{1\}\}\\*
 =&\;\{\{1,2\}\}.
\end{align*}
It is shown analogously that $\delta\{\emptyset\} = \emptyset$. This entails that 
$\delta\{X_i\}_{i\in I} = \emptyset$ whenever $X_i = \emptyset$ for some $i\in I$ --
otherwise we would obtain a contradiction to naturality of $\delta$ w.r.t.\ the constant 
function $f\c X\to\{1\}$. This completes the proof that the only distributive law
satisfying~\eqref{eq:barr-law} is~\eqref{eq:barr-ne}.

Now, let us assume~\eqref{eq:em-law}. Let $i\in I$ and let $x\in X_i$. Let $f\c X\to\{0,1\}$
send $x$ to $1$ and all the other elements to $0$. Then $(\ftP\ftF f) (\delta \{X_i\}_{i\in I}) = \delta (\ftF\ftP f)\{X_i\}_{i\in I}$
is either $\delta\{\{0,1\},\{0\}\} = \{\{0,1\}\}$ or $\delta\{\{1\},\{0\}\} = \{\{0,1\}\}$
or $\delta\{\{1\}\} = \{\{1\}\}$. In all cases, the result does not contain $\{0\}$,
which would happen if some $Y_j$ did not contain $x$. We thus proved that for every 
$i\in I$ and every $j\in J$, $X_i\subseteq Y_j$. Therefore~\eqref{eq:PP-dl} turns 
into 
\begin{align*}
\delta\{X_i\}_{i\in I} = \Bigl\{\bigcup_{i\in I} X_i\Bigr\}
\end{align*}
By \eqref{eq:d-e-ee}, the same formula is still valid if $I=\emptyset$. Let us show that 
the same formula is also valid for $|I|>1$ and $X_i = \emptyset$ for some $i$.
With $I'\neq\emptyset$ indeed, we have
\begin{align*}
\delta(\{X_i\}_{i\in I'}\cup\{\emptyset \})
 =&\;(\mu (\ftP\delta))\{\{X_i\}_{i\in I'}\cup\{\emptyset \}\}\\*
 =&\;(\mu (\ftP\delta)\delta)\{\{X_i,\emptyset \}\mid i\in I'\}\\
 =&\;(\delta (\ftF\mu))\{\{X_i,\emptyset \}\mid i\in I'\}\\*
 =&\;\delta\{X_i\}_{i\in I'}\\*
 =&\;\Bigl\{\bigcup\nolimits_{i\in I'} X_i\Bigr\}.
\end{align*}
The only case that is still not covered is $\{X_i\}_{i\in I} = \{\emptyset\}$. 
For $\delta\{\emptyset\}$, only two choices are possible $\delta\{\emptyset\} = \{\emptyset\}$
and $\delta\{\emptyset\} = \emptyset$. The former option is a special case of \eqref{eq:em-law}.
The latter option still gives a distributive law, which is verified directly.\qed
\end{proof}

\section{Conclusions and Future Work}

We have studied existence and uniqueness of distributive laws of set
functors over the powerset monad -- equivalently, of extensions of
functors from $\SET$ to $\REL$. We have identified the class of
\emph{elementwise bounded} set functors, which strictly includes all
accessible functors. For elementwise bounded functors, a distributive
law exists if and only if weak pullbacks are preserved, and in that
case the distributive law is unique and given by the power law
(\Cref{p:13}), which induces the Barr extension.  Our uniqueness
argument consisted in showing that every extension of an elementwise
bounded set functor is \emph{locally monotone}, revealing an
unexpected connection between local monotonicity of extensions and
accessibility-like conditions on set functors.  The powerset functor
itself (which fails to be elementwise bounded) shows that uniqueness
may fail beyond this setting: there are exactly two additional
distributive laws of the full powerset functor over the (full)
powerset monad (\Cref{pro:3-laws}).

Overall, our work clarifies the degree of canonicity of extensions of
functors to relations. One important direction for future work is to
obtain similar clarifications for real-valued relations and, more
generally, for quantitative relations, where the overall picture must
be expected to be rather more complex.

\bibliographystyle{splncs04}
\bibliography{references}
\end{document}